\documentclass[10pt,aip,cha,
 				twocolumn,reprint,
 				superscriptaddress,numerical]{revtex4-1}

\usepackage{amsmath,amsthm}
\usepackage{amssymb}
\usepackage{amscd}
\usepackage{wasysym}
\usepackage[ansinew]{inputenc}
\usepackage[T1]{fontenc}
\usepackage{ae,aecompl}
\usepackage{sidecap}
\usepackage{hyperref}

    \usepackage[pdftex]{graphicx}
    \DeclareGraphicsExtensions{.pdf}

\usepackage{color}
\definecolor{red}{rgb}{1,0,0}
\definecolor{lred}{rgb}{1,0.7,0.7}
\definecolor{lgreen}{rgb}{0.7,1,0.7}

\newcommand{\be}{\begin{equation}}
\newcommand{\ee}{\end{equation}}
\newcommand{\bea}{\begin{eqnarray}}
\newcommand{\eea}{\end{eqnarray}}
\newcommand{\nn}{\nonumber}

\newcommand{\tr}{{\rm tr}}
\renewcommand{\vec}[1]{\boldsymbol{#1}}

\newtheorem{theorem}{Theorem}
\newtheorem{definition}{Definition}
\newtheorem{lemma}{Lemma}

\begin{document}
\bibliographystyle{apsrev}

\title{A universal order parameter for synchrony in networks of limit cycle oscillators}

\author{Malte Schr\"oder}
\affiliation{Network Dynamics, Max Planck Institute for Dynamics and Self-Organization (MPIDS), 37077 G\"ottingen, Germany}

\author{Marc Timme}
\affiliation{Network Dynamics, Max Planck Institute for Dynamics and Self-Organization (MPIDS), 37077 G\"ottingen, Germany}
\affiliation{Network Dynamics, Technical University of Dresden, Institute for Theoretical Physics, 01062 Dresden, Germany}

\author{Dirk Witthaut}
\affiliation{Forschungszentrum J\"ulich, Institute for Energy and Climate Research -
	Systems Analysis and Technology Evaluation (IEK-STE),  52428 J\"ulich, Germany}
\affiliation{Institute for Theoretical Physics, University of Cologne, 
		50937 K\"oln, Germany}

\date{\today }

\begin{abstract}

We analyze the properties of order parameters measuring synchronization and phase locking in complex oscillator networks. First, we review network order parameters previously introduced and reveal several shortcomings: none of the introduced order parameters capture all transitions from incoherence over phase locking to full synchrony for arbitrary, finite networks. We then introduce an alternative, universal order parameter that accurately tracks the degree of partial phase locking and synchronization, adapting the traditional definition to account for the network topology and its influence on the phase coherence of the oscillators. We rigorously proof that this order parameter is strictly monotonously increasing with the coupling strength in the phase locked state, directly reflecting the dynamic stability of the network. Furthermore, it indicates the onset of full phase locking by a diverging slope at the critical coupling strength. The order parameter may find applications across systems where different types of synchrony are possible, including biological networks and power grids.
\end{abstract}

\maketitle

\begin{quotation}
Many dynamical system in physics, biology or engineering can be described as coupled phase oscillators, often in a network with a complex interaction topology. The prototypical model considered in this context are networks of Kuramoto oscillators. To study the synchronization in such systems several order parameters have been introduced, adapting the original Kuramoto order parameter, defined for all-to-all coupled oscillators, to complex interaction networks. However, none of the order parameters manages to fully track the transition from oscillators moving at their individual frequencies to full synchronization of the network. Here we propose a universal order parameter to study synchronization in finite networks of phase oscillators, tracking all stages of synchronization. This order parameter may be used to study systems where different stages of synchrony are relevant. Additionally, we rigorously proof several helpful qualities relating the order parameter not only to the synchrony but also to the  dynamical stability of the network. 
\end{quotation}

\section{Introduction}

Many oscillatory systems enter stable limit cycles as their dynamic steady state. If such systems are coupled, they often interact only through their positions along their periodic orbit, their phases. The simplest prototypical model to describe such coupled phase oscillators is the celebrated Kuramoto model \cite{Kura75,Stro00}. It characterizes the collective dynamics of a variety of phase oscillator systems ranging from chemical reactions \cite{Kura84} and neural networks \cite{Somp90,kirst16_dynamic} to coupled Josephson junctions \cite{Wies96}, laser arrays \cite{Vlad03}, optomechanical systems \cite{Hein11} and mean-field quantum systems \cite{14Witthaut,17Witthaut}.

Studies of the Kuramoto model and more general phase oscillator networks typically focus on the onset of synchronization between the individual oscillators \cite{Kura75,Stro00,Kura84,Aceb05,Doerfler14}. Starting from the analytical results for the mean field behavior in the all-to-all coupled Kuramoto model, correctly predicting the emergence of partial phase locking, extensions of this result to various network topologies were developed \cite{Timm06,06Bocaletti,07Gomez,08Arenas}. These extensions often use a similar methodology and define an adapted order parameter to analyze the transition to synchrony. Interestingly, none of these order parameters captures all transitions from the incoherent to the completely synchronized state for arbitrary, finite networks.

Depending on the application different states of phase ordering are relevant and a different order parameter is appropriate. Commonly, the onset of partial phase locking has received most interest \cite{Kura75,Stro00,Kura84}. For example, partial phase locking indicates the growth of number fluctuations in quantum mean-field models \cite{14Witthaut,17Witthaut}. In contrast, in technical systems such as power grids, a fully phase locked state is required for stable operation \cite{12powergrid,Dorf13,Mott13,16critical}.

We propose a universal order parameter that accurately reflects the phase coherence of phase oscillators in any network, describing the initial growth of partially phase locked clusters as well as the convergence to full synchrony. This order parameter is particularly suited to study the fully phase locked state as it directly reflects the dynamic stability of this steady state. It increases monotonically with the coupling strength, in contrast to previously defined mean field order parameters.

\begin{figure*}
\centering
\includegraphics[width = 0.9\textwidth]{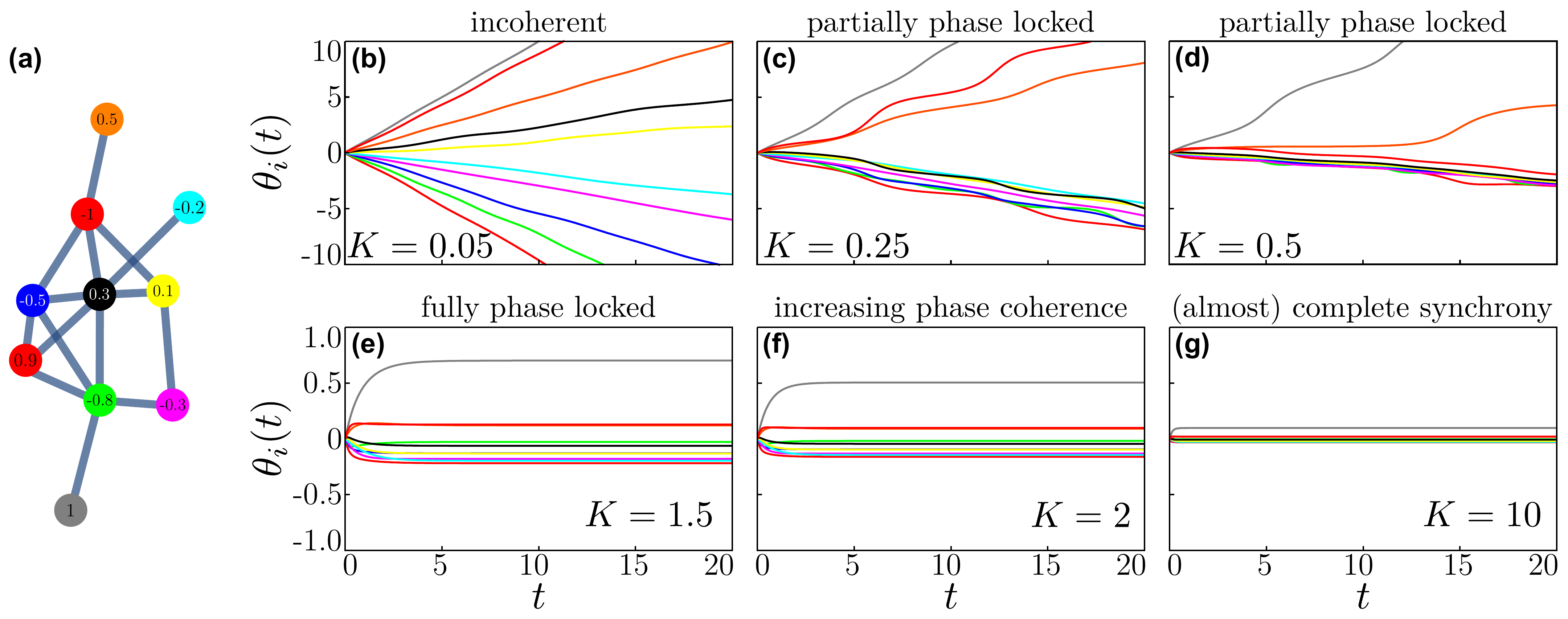}
\caption{\textbf{Synchronization in the Kuramoto model.}
Dynamics of the Kuramoto model for $N=10$ oscillators with a random interaction network. The phase coherence between neighboring oscillators increases with the coupling strength, eventually leading to full synchrony of all oscillators. (a) Topology of the interaction network, the numbers denote $\omega_i$ of the respective oscillator. (b) For small coupling the oscillators move (almost) independently with their individual frequencies (slope). (c,d) As the coupling strength increases beyond $K_{c1} = 0.1$ some oscillators enter a partially phase locked state and their phases evolve with the same time-averaged frequency. (e) If the coupling strength becomes larger than $K_{c2} = 1$, all nodes are phase locked and move with the same constant frequency $\frac{d \theta_i}{dt} = 0$. (f,g) Further increasing the coupling strength reduces the phase differences of the oscillators until complete synchrony $\theta_i - \theta_j = 0$ is achieved for $K \rightarrow \infty$.}
\label{fig:kuramoto_example}
\end{figure*}

\section{Phase oscillators and the Kuramoto model}

Limit cycles are ubiquitous as dynamically stable states in a wide range of systems. When such systems are coupled, interactions can typically be approximated as interactions between their phases $\theta_i$. The Kuramoto model
\be
  \frac{d \theta_i}{dt} = \omega_i + 
     K \sum_{j=1}^N A_{i,j}  \sin(\theta_j - \theta_i)
     \label{eqn:kuramoto-intro}
\ee
is one of the simplest models for such coupled phase oscillators. It describes the dynamics of $N$ oscillators with natural frequencies $\omega_i$ and sinusoidal coupling. The parameter $K$ denotes the coupling strength of the interactions and $A_{i,j} \in \left\{0,1\right\}$ is the adjacency matrix of the interaction network, describing which nodes interact with which other nodes. The results easily extend to inhomogeneous coupling strengths with $A_{i,j} \in \mathbb{R}$. In many applications, interactions between individual oscillators are reciprocal and in the following we assume an undirected network, i.e., a symmetric adjacency matrix $A_{i,j} = A_{j,i}$. Similarly, we can without loss of generality consider a co-rotating frame such that the natural frequencies of the oscillators are centered around $0$ and we have $\sum_i \omega_i = 0$, where the sum runs from $1$ to $N$. In the following we only consider connected networks, as otherwise we can treat the connected sub-systems individually.

The dynamics of coupled Kuramoto oscillators depends strongly on the strength $K$ of the interactions. For small coupling $K$ all oscillators rotate (almost) independently with their natural frequencies $\omega_j$. In this state the phases are \emph{incoherent}. Above some critical coupling strength $K \ge K_{c1}$, a subset of the oscillators starts to synchronize such that their time averaged frequencies $\left< \frac{d \theta_i}{dt} \right>_t$ become identical. The phases of these oscillators then move together in a \emph{partially phase locked} state and their phase differences $\theta_i - \theta_j$ are bounded. When the coupling becomes even stronger, $K \ge K_{c2}$, a \emph{fully phase locked} state appears in a saddle-node bifurcation \cite{14bifurcation}. All oscillators synchronize to a common frequency $\frac{d \theta_i}{dt} = \mathrm{const.} = 0$ and the phase differences between all nodes become constant $\theta_i - \theta_j = \mathrm{const}$. Further increasing the coupling reduces the phase differences until \emph{complete synchronization} of the oscillators, defined by $\theta_i - \theta_j = 0$, is achieved as $K \rightarrow \infty$. This behavior is illustrated in Fig.~\ref{fig:kuramoto_example} showing the dynamics of a small random network of oscillators for various coupling strengths.

Most studies focus on the transition from incoherent oscillators moving at their individual frequencies to a partially phase locked state \cite{Kura75,Stro00,Kura84,Aceb05,Doerfler14}. In a variety of technical systems, however, partial phase coherence is not sufficient for stable function. For instance, Kuramoto-like dynamics appear in a second order model describing the frequency dynamics of power grids \cite{12powergrid,12braess,13nonlocal,14bifurcation,schafer2015decentral,16critical,Dorf13,Mott13}:
\begin{equation}
	M_i \frac{\mathrm{d}^2\theta_i}{\mathrm{d}t^2} - D_i \frac{\mathrm{d}\theta_i}{\mathrm{d}t}
	 = P_i  + \sum_{j = 1}^{N} K A_{i,j}  \sin(\theta_j - \theta_i).
	 \label{eq:powergrid}
\end{equation}
Here, $M_i$ is the inertia, $D_i$ the damping coefficient and $P_i$  the power injection at node $i$. The phases $\theta_i(t)$ describe the state of rotating machines (generators or motors) and the coupling their interactions via power transmission lines. In the steady state $\frac{\mathrm{d}\theta_i}{\mathrm{d}t} = 0$, required for stable operation of the power grid, all machines work at the same frequency. This state is characterized by the same equations that describe a fully phase locked state in the Kuramoto model. The stability of this state and how the phase cohesiveness in the network scales with the coupling strength is an important question \cite{Doerfler10}.

Ideally, a universal order parameter would be able to characterize both the transition to partial as well as to full phase locking and the properties of a phase locked state in arbitrary, especially finite networks.

\section{Kuramoto order parameters}

To quantitatively study the transitions from an incoherent to a fully synchronous state one typically introduces an order parameter to measure the phase coherence. For the original all-to-all coupling model, Kuramoto introduced the complex order parameter \cite{Kura84,Stro00}
\be
      r(t)e^{\text{i}\psi(t)} = \frac{1}{N} \sum_{i=1}^N e^{\text{i}\theta_i} \,,
      \label{eqn:def-order-r}
\ee
where $\psi(t)$ describes the average phase of all oscillators and $r(t)$ the degree of phase coherence.
A single measure for the phase ordering is the given by the long time average of the absolute value of the order parameter
\bea
      r^2_\mathrm{Kuramoto} &=& \left< \left|r(t)e^{\text{i}\psi(t)}\right|^2 \right>_t = \left< r(t)^2 e^{\text{i}\psi(t)} e^{-\text{i}\psi(t)}\right>_t \nonumber\\
      &=& \left< \frac{1}{N^2} \sum_{i,j = 1}^N e^{\text{i}(\theta_i - \theta_j)} \right>_t \nonumber\\
      &=& \left< \frac{1}{N^2} \sum_{i,j = 1}^N \cos(\theta_i - \theta_j) \right>_t\,.
      \label{eqn:def-order-r2}
\eea
This order parameter measures the average of the phase differences of all pairs of oscillators. If the oscillators are incoherent, the time average vanishes and the order parameter is $0$. When a fraction of the oscillators are partially phase locked the cosine of their phase differences becomes positive and does not disappear in the time average; the order parameter becomes positive.

In the original case for $N$ all-to-all coupled oscillators with natural frequencies $\omega_i$ following a distribution $g(\omega)$, mean-field theory correctly predicts the transition to partial phase coherence at the critical coupling $K_{c1} = 2/\left[\pi g(0)\right]$ if the frequency distribution $g$ is unimodal and symmetric around zero. For larger coupling strengths $K>K_{c1}$ the order parameter then grows continuously as $r(K) \propto \sqrt{1-K_{c1}/K}$ \cite{Stro00}. As such, this order parameter characterizes the transition from an \emph{incoherent} to a \emph{partially} phase locked state. 

This original order parameter is clearly unsuited when studying more general interaction networks. One would compare the phases of two oscillators in the network that are only interacting indirectly via a (possibly very long) chain of intermediate oscillators. As such, several adaptations of the order parameter have been introduced to study the effect of the network topology on the synchronization of Kuramoto oscillators:

The first definition used by Restrepo et al. \cite{05Restrepo,06Restrepo,08Arenas} considers an intuitively defined local order parameter
\be
	r_i = \left|\sum_{j=1}^N  A_{i,j} \left<e^{\text{i}\theta_j}\right>_t\right|
\ee
for oscillator $i$, measuring the phase coherence of all neighboring oscillators. A global order parameter is then easily defined as the average of the local order parameters
\be
	r_\mathrm{net} = \frac{\sum_{i=1}^N r_i}{\sum_{i=1}^N k_i},
\ee
where $k_i$ is the degree of node $i$. 

A second definition \cite{04Ichinomiya,06Bocaletti} adapts the original order parameter Eq.~(\ref{eqn:def-order-r}) weighting each node with its degree
\be
	r_\mathrm{mf} = \left<\left|\frac{\sum_{i=1}^N k_i e^{\text{i}\theta_i}}{\sum_{i=1}^N k_i}\right|\right>_t .
\ee
This order parameter ignores the specific network topology in favor of a mean-field view of network ensembles to simplify analytical calculations.

Finally, a definition of an order parameter to study local synchronization used in \cite{07Gomez} derives from the original order parameter Eq.~(\ref{eqn:def-order-r2}), restricting it to the network topology and only averaging over the phase differences between directly connected nodes
\be
	r_\mathrm{link} = \frac{1}{\sum_{i=1}^N k_i} \sum_{i,j = 1}^N A_{i,j} \left| \left< e^{\text{i} \left(\theta_i - \theta_j\right)} \right> _t \right| \,.
\ee

The above order parameters work well for their respective use, for example to study synchronization analytically in mean-field network models. However, none of them accurately captures the whole transition to synchronization, especially in smaller networks. We illustrate this in Fig.~\ref{fig:network_order_params} for a small random network: While $r_\mathrm{net}$ clearly captures the transition to full phase locking at $K_{c2}=1$, it is effectively $0$ before full phase locking becomes stable and does not indicate where individual nodes enter the partially phase locked state for $K < 1$. Conversely, $r_\mathrm{link}$ describes these transitions but cannot cover the convergence to full synchrony as  $r_\mathrm{link} = 1$ in the fully phase locked state, regardless of the network topology. Finally, $r_\mathrm{mf}$ works well to describe the behavior for a large ensemble of networks, but is clearly unsuited for use with specific, particularly small, networks as it ignores the specific network structure and is large already for weak coupling. It is easy to construct further examples where, for instance, the mean field order parameter $r_\mathrm{mf}$ is non-monotonous with respect to the coupling strength $K$, even in the fully phase locked state.

\begin{figure*}
\centering
\includegraphics[width = 0.7\textwidth]{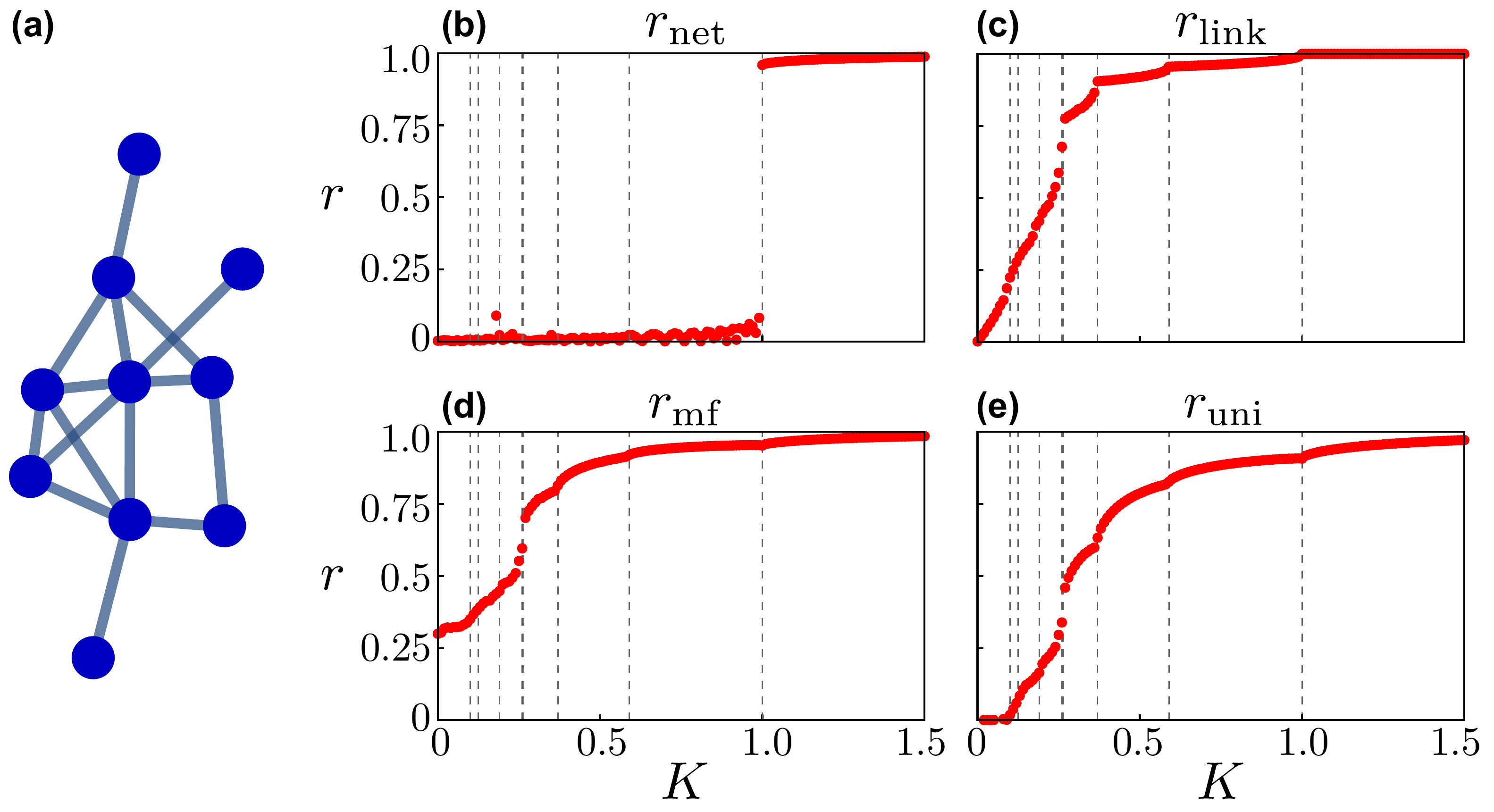}
\caption{\textbf{Order parameters to measure phase coherence in networks.}
Different order parameters measuring the phase coherence in complex networks of Kuramoto oscillators, describing the transition from a completely incoherent state [$K=0$, cf. Fig.~\ref{fig:kuramoto_example}(b)] to full synchrony [$K\rightarrow \infty$, cf. Fig.~\ref{fig:kuramoto_example}(g)]. None of the order parameters used in the literature $r_\mathrm{net}$, $r_\mathrm{mf}$ and $r_\mathrm{link}$ captures all transitions. (a) Topology of the interaction network, cf. Fig.~\ref{fig:kuramoto_example}(a). (b) $r_\mathrm{net}$ is almost $0$ until the fully phase locked state becomes stable at $K=1$. It fails to capture transitions in the partially phase locked regime. (c) In contrast, $r_\mathrm{link}$ captures the transitions in the partially phase locked regime well. However, $r_\mathrm{link} = 1$ in the fully phase locked state for $K \ge 1$ and does not capture the convergence to complete synchrony. (d) $r_\mathrm{mf}$ measures globally averaged phase coherence. It fails to accurately represent the incoherent and partially phase locked state with respect to the actual network topology, especially for small networks. (e) Our universal order parameter $r_\mathrm{uni}$ accurately reflects the degree of phase coherence in all stages of synchronization. All results show the long time limit of the order parameter starting from identical initial conditions $\theta_i = 0$, the black dashed lines mark transitions where single nodes enter a (partially) phase locked state.}
\label{fig:network_order_params}
\end{figure*}

\section{A universal order parameter for complex networks}

In order to have both a practically applicable and relevant order parameter as well as describe the whole evolution from an incoherent state to complete synchronization we propose a universal network order parameter:
\begin{definition}
	Given a network of coupled Kuramoto oscillators Eq.~(\ref{eqn:kuramoto-intro}), phase ordering is measured by
\bea
	r_\mathrm{uni} &=& \frac{1}{\sum_{i=1}^N k_i} \sum_{i,j=1}^N A_{i,j} \left< \Re\left( e^{\text{i}\left(\theta_i - \theta_j\right)} \right)\right>_t \nonumber\\
				   &=& \frac{1}{\sum_{i=1}^N k_i} \sum_{i,j=1}^N A_{i,j} \left< \cos\left(\theta_i - \theta_j\right) \right>_t . \label{eqn:order_rho}
\eea
\end{definition}

As $r_\mathrm{link}$ this definition respects the topology of the interaction network and considers only phase differences between neighboring nodes. In contrast to $r_\mathrm{link}$, the definition of $r_{\rm uni}$ reduces to the original Kuramoto order parameter Eq.~(\ref{eqn:def-order-r}) for a completely connected network as desired. Figure~\ref{fig:network_order_params}(d)  illustrates the behavior in comparison to the other network order parameters, showing that it accurately captures the transitions in all stages of phase locking (cf. Fig.~\ref{fig:network_order_table}).

\subsection{Synchronization and stability}

The order parameter $r_{\rm uni}$ gives a full account of the emergence of synchrony. It accurately follows both the transitions to partially and fully phase locked states as well as the convergence to complete synchrony. 

We illustrate this central result in Fig.~\ref{fig:network_order_params} for a small random network. Whenever one of the nodes enters a partially phase locked state we observe a strong kink in $r_{\rm uni}(K)$. Hence, we can directly track the growth of phase locked clusters. In fact, the slope $\mathrm{d} r_{\rm uni}/\mathrm{d} K$ diverges when approaching these transition points from the right. We rigorously proof this result for the transition to full phase locking below (cf. Theorem 1).

The universal order parameter has further advantages compared to the alternatives discussed above. First, $r_{\rm uni}$ quantifies the dynamical stability of a phase-locked steady state (cf. Theorem 2).
This becomes most apparent in a ring of $N$ oscillators with identical natural frequencies, $\omega_i = 0$ for all $i \in \left\{1,2,\dots,N\right\}$, where all interactions have identical coupling strength $K = 1$. Clearly, in a fully phase locked state all phase differences between neighboring nodes need to be identical while the cumulative phase difference around the ring must be a multiple of $2 \pi$ \cite{17Manik,delabays2016multistability}. Under these conditions we can characterize the phase locked states by a mode $m$ describing the total phase change around the ring $2 \pi m$. The individual phases are then given by 
\be
   \theta_i^* = \frac{2 \pi i m}{N} \,
\ee 
with $m \in \left\{ -N/2, -N/2 + 1, \dots,N/2 \right\}$, illustrated for $m \ge 0$ in Fig.~\ref{fig:order_stability}. Here and in the following we use an asterisk $*$ to denote a phase locked steady state $\theta_i^*$ of the Kuramoto model Eq.~(\ref{eqn:kuramoto-intro}).

The phase locked states with $\left|\theta_i^* - \theta_{i-1}^* \right| < \pi / 2$, that means $m\in (-N/4,N/4)$, are linearly stable, the remaining states are unstable. Our order parameter $r_{\rm uni}$ reflects the linear stability of these different steady states - the state with perfectly aligned phases ($m=0$) is most stable and has $r_{\rm uni}= 1$. All other states have larger phase differences, which impede dynamical stability, and consequently lower values of $r_{\rm uni}$. This information is completely lost for the alternatives $r_{\rm link}$ and $r_{\rm mf}$, the first one being identically one for all phase-locked states and the second one being one for the fully aligned state and zero otherwise.

The classification of stability is due to the fact that $r_\mathrm{uni}$ Eq.~(\ref{eqn:order_rho}) counts only the phase differences in the stable region as positive contributions, i.e., when \mbox{$\left| \theta_i^* - \theta_j^* \right| < \pi/2$}. As the stability of phase locked state is directly related to these phase differences, with phase differences close to $0$ corresponding to more stable states, the order parameter directly reflects the systems stability of any phase locked state, relevant for example for applications to power grids. 

A further advantage of $r_{\rm uni}$ for the analysis of phase-locked states is monotonicity (cf. Theorem 2). Intuitively we expect that an increase of the coupling $K$ leads to a stronger alignment of the phases and thus to an increase of the order parameter. This expectation can be violated for the mean-field order parameter $r_{\rm mf}$,  as it  measures global alignment, but an increase of the coupling acts only locally on the links. In contrast, we rigorously proof below that the order parameter $r_{\rm uni}$ is monotonic in the coupling strength $K$ for a phase-locked steady state.

\begin{figure}
\centering
\includegraphics[width = 0.45\textwidth]{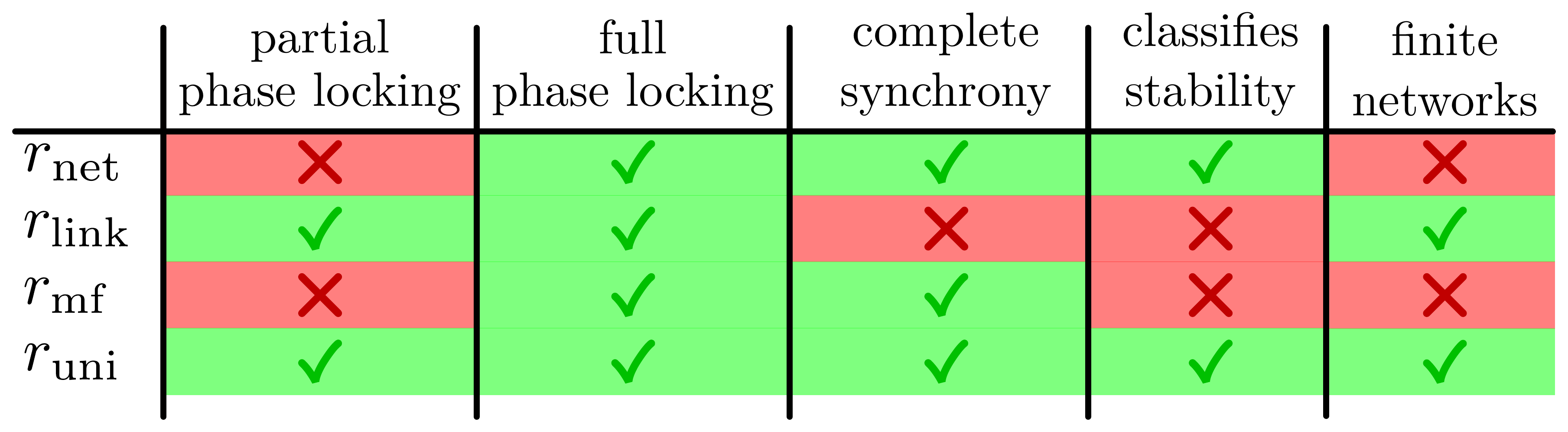}
\caption{\textbf{A universal order parameter.}
None of the order parameters used in the literature $r_\mathrm{net}$, $r_\mathrm{mf}$ and $r_\mathrm{link}$ capture all transitions. Following the observations in Fig~\ref{fig:network_order_params}, $r_\mathrm{net}$ fails to capture transitions in the partially phase locked regime. It also fails to describe phase coherence for some small networks, most easily seen for just two connected oscillators. $r_\mathrm{link}$ does not capture the transition to complete synchrony and, since $r_\mathrm{link} = 1$ in the fully phase locked state, it does not classify stability. $r_\mathrm{mf}$ does not reflect the phase ordering in networks for partially or fully phase locked states, since it measures global phase coherence. As such it does not represent stability of the phase locked steady states which depends on local phase differences and is not suited for small networks. The order parameter $r_\mathrm{uni}$ accurately reflects the transitions for all stages of synchronization and correctly classifies stability of different phase locked states in arbitrary, even small networks.}
\label{fig:network_order_table}
\end{figure}

\subsection{Analytical results}

To formalize these observations, first consider the linear stability of a phase locked state $\vec \theta^*$ for $K \ge K_{c,2}$:
A small perturbation $\vec \xi$ around the steady state, $\theta_j = \theta_j^* + \xi_j$, evolves as    
\be
  \frac{\mathrm{d}}{\mathrm{d}t} \vec \xi = \vec J \vec \xi + \mathcal{O}(\vec \xi^2), \label{eqn:linear_stability}
\ee
where we make use of vector notation $\vec \xi = (\xi_1,\ldots,\xi_N)^T$. The Jacobian matrix $\vec J$ quantifies the linear stability of a phase-locked steady state. It always has one trivial eigenvalue $\lambda_1 = 0$ with eigenvector $\vec v_1 = \left(1,1,\ldots,1\right)^\mathrm{T}$, representing a global uniform shift of all phases which does not affect the phase-locking of the nodes. In a stable phase locked state all other eigenvalues are negative $0 > \lambda_2 \ge  \lambda_3 \ge \cdots \ge \lambda_N$. We denote the associated eigenvectors as $\vec v_{2}, \ldots, \vec v_N$.

We can then formalize the above observations about $r_{\rm uni}$ in the following theorems:

\begin{theorem}
Given a network of coupled Kuramoto oscillators Eq.~(\ref{eqn:kuramoto-intro}) with $\sum_i \omega_i = 0$ and $\vec \omega \cdot \vec v_2 \neq 0$, the derivative of the order parameter $r_\mathrm{uni}$ Eq.~(\ref{eqn:order_rho}) diverges when the fully phase locked state becomes unstable at the critical coupling $K_{c2}$
\be 
	\mathrm{d}r_\mathrm{uni}/\mathrm{d}K \rightarrow \infty \quad\mathrm{for} \quad K \rightarrow K_{c2}^+\,. \nn
\ee
\end{theorem}

\begin{theorem}
Given a network of coupled Kuramoto oscillators Eq.~(\ref{eqn:kuramoto-intro}) with $\sum_i \omega_i = 0$, in a fully phase locked regime $K > K_{c2}$ the order parameter $r_\mathrm{uni}$ Eq.~(\ref{eqn:order_rho}) is strictly larger than zero for every stable phase locked state and increases monotonically with increasing $K$.
\end{theorem}

\begin{figure*}
\centering
\includegraphics[width = 0.7\textwidth]{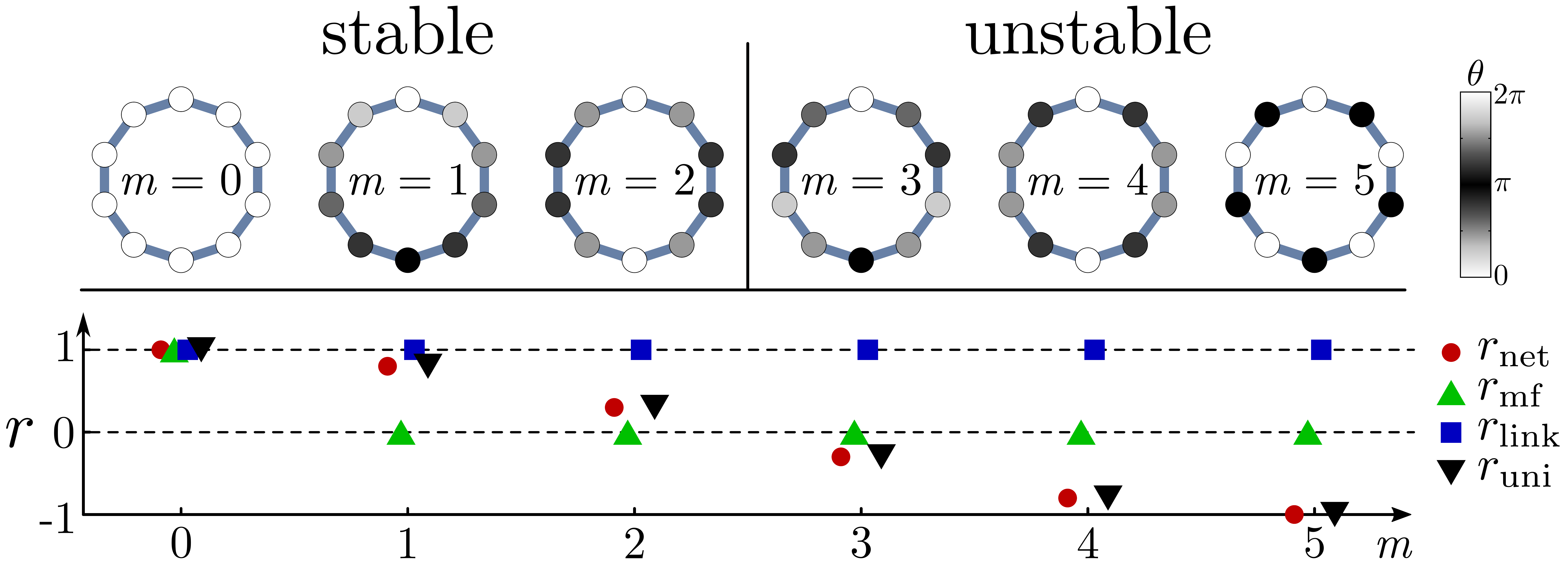}
\caption{\textbf{Order parameters and stability.}
Steady states in a ring network with $N=10$ nodes and the corresponding values of the different order parameters (shifted horizontally for better visibility). The state $m=0$ is the most stable as the phase differences between neighboring nodes are smallest. The phase locked states become more unstable with increasing $m$. $r_\mathrm{mf}$ and $r_\mathrm{link}$ do not provide information about the stability of the steady state, being either zero for most of the states or identical to one for all phase locked states, respectively. Our universal order parameter $r_\mathrm{uni}$ accurately reflects the stability of the different states.}
\label{fig:order_stability}
\end{figure*}

In the remainder of this section we provide the proof for these theorems with the help of two lemmas, relating the order parameter to the eigenvalues of the Jacobian:

\begin{lemma}
Given a network of coupled Kuramoto oscillators Eq.~(\ref{eqn:kuramoto-intro}) with $\sum_i \omega_i = 0$ and $K \ge K_{c2}$ in the stable phase locked state, the order parameter $r_\mathrm{uni}$ Eq.~(\ref{eqn:order_rho}) is given by the negative trace of the Jacobian $\vec J$,
\bea
    r_{\rm uni} &=& - \frac{1}{K \sum_{i=1}^N k_i} \,  \tr({\vec J}) \nn \\
               &=&  - \frac{1}{K \sum_{i=1}^N k_i} \sum_{j=2}^N \lambda_j.
\eea
\end{lemma}

\begin{proof}
Explicit calculation of the Jacobian matrix $\vec J$ in Eq.~(\ref{eqn:linear_stability}) yields
\begin{align}
    J_{i,j} &= K A_{i,j}  \cos(\theta_i^* - \theta_j^*)  \qquad \mbox{for} \;  i \neq j, \nn \\
    J_{i,i} &= - K \sum_{j=1}^N A_{i,j} \cos(\theta_i^* - \theta_j^*).
    \label{eqn:def-jacobian}
\end{align}
The lemma then follows directly by calculating the trace. The second equality follows from the fact that the largest eigenvalue of $\vec J$ is $\lambda_1 = 0$.
\end{proof}

Given that the eigenvalues of the Jacobian $\lambda_2, \dots, \lambda_N < 0$ are all negative for a stable phase locked state, $K > K_{c2}$, it immediately follows that the order parameter $r_\mathrm{uni}$ must be positive.

To finish proving the theorems above, we now also relate the derivative $\mathrm{d}r_\mathrm{uni}/\mathrm{d}K$ to the eigenvalues $\lambda_1,\dots,\lambda_N$ of the Jacobian matrix and their corresponding eigenvectors $\vec v_1,\ldots,\vec v_N$:

\begin{lemma}
Given a network of coupled Kuramoto oscillators Eq.~(\ref{eqn:kuramoto-intro}) with $\sum_i \omega_i = 0$ and $K \ge K_{c2}$ the derivative of the order parameter with respect to the coupling strength
is given by 
\bea
   \frac{\mathrm{d} r_{\rm uni}}{\mathrm{d} K} =  
    \frac{2}{K^2 \sum_{i=1}^N k_i} \sum_{n=2}^N  \frac{1}{-\lambda _n}
               (\vec v_n \cdot \vec \omega )^2 \ge 0.
    \label{eq:slope-r}           
\eea 
\end{lemma}

\begin{proof}
Consider a global change of the coupling strength $K' = K + \kappa$.  This perturbation induces a small change of the steady state phases of the network, $\theta_m^* \rightarrow \theta'_m = \theta_m^* + \xi_m$.
Expanding the steady state condition 
\bea
    0 = \omega_i + (K+\kappa) \sum_{m = 1}^{N} A_{i,m}  \sin(\theta_m^* + \xi_m - \theta_i^* - \xi_i) \nn
\eea
to leading order in $\kappa$ and the $\xi_m$ yields
\bea
  && 0 = \kappa \sum_{m = 1}^N A_{i,m} \sin(\theta_m^* - \theta_i^*)   
       + \sum_{m =1}^N J_{i,m} \xi_m      \nn \\
   \Rightarrow \, &&   \sum_{m =1}^N J_{i,m} \xi_m  = - \frac{\kappa}{2}
       \sum_{\ell,m=1}^N A_{\ell,m} \sin(\theta^*_m - \theta^*_\ell) (\delta_{i,\ell} - \delta_{i,m})  \nn
   \label{eqn:steady2}
\eea
for all $i \in \left\{ 1,\ldots,N \right\}$ using the definition of the Jacobian Eq.~(\ref{eqn:def-jacobian}) and the Kronecker $\delta$ symbol. In vectorial notation this set of equations can be written as
\be
   \vec J \vec \xi = - \frac{\kappa}{2} \sum_{\ell,m=1}^N A_{\ell,m} \sin(\theta^*_m - \theta^*_\ell)
        \vec q_{(\ell,m)},
    \label{eq:Jxi}    
\ee
where we define the vector $\vec q_{(\ell,m)}$, whose $i$th component is given by $\vec q_{(\ell,m),i} = \delta_{i,\ell} - \delta_{i,m}$. The matrix $\vec J$ is singular, but the vectors $\vec q_{(\ell,m)}$ are orthogonal to its kernel [$\vec v_1 = \left(1,1,\dots,1\right)^T$] such that we can solve equation (\ref{eq:Jxi}) using the Moore-Penrose pseudo-inverse $\vec J^+$. Decomposing $\vec J$ into eigenvalues and eigenstates, we thus obtain 
\bea
   \vec \xi &=& - \frac{\kappa}{2} \sum_{\ell,m=1}^N A_{\ell,m} \sin(\theta^*_m - \theta^*_\ell)
       \vec J^+ \vec q_{(\ell,m)}   \nn \\
    &=& - \frac{\kappa}{2} \sum_{\ell,m=1}^N \sum_{n=2}^N A_{\ell,m} \sin(\theta^*_m - \theta^*_\ell)
         \frac{1}{\lambda_n} (\vec v_n \cdot \vec q_{(\ell,m)}) \vec v_n. \nn
\eea

We then find for the change of the phases 
\bea
    \frac{\mathrm{d}(\theta_j - \theta_i)}{\mathrm{d}K} 
    &=& \vec q_{(j,i)} \cdot \lim_{\kappa \rightarrow 0} 
       \underbrace{\frac{\vec \theta(K+\kappa)-\vec \theta(K)}{\kappa}}_{= \vec \xi/\kappa } \nn \\
    &=& -\frac{1}{2} \sum_{\ell,m=1}^N  A_{\ell,m} \sin(\theta^*_m - \theta^*_\ell) \nn \\
    && \qquad \qquad \times \sum_{n=2}^N \frac{1}{\lambda_n} (\vec q_{(\ell,m)} \cdot \vec v_n) (\vec q_{(j,i)} \cdot \vec v_n). \nn 
\eea
Hence, the derivative of the order parameter is given by
\bea
  &&  \frac{\mathrm{d} r_{\rm uni}}{\mathrm{d} K} =
      \frac{1}{\sum_{i=1}^N k_i}  \sum_{i,j = 1}^N A_{i,j} 
            \frac{d\cos(\theta_i^* - \theta_j^*)}{dK} \nn \\
      && \; =  \frac{1}{\sum_{i=1}^N k_i}  \sum_{i,j = 1}^N A_{i,j} 
            \sin(\theta_i^* - \theta_j^*) \frac{d(\theta_j - \theta_i)}{dK} \nn \\      
   && \; =  \frac{1}{2\sum_{i=1}^N k_i} \sum_{n=2}^N  \frac{1}{-\lambda _n}
               \left[  \sum_{i,j = 1}^N A_{i,j}  \sin(\theta_i^* - \theta_j^*)
                 (\vec q_{(j,i)} \cdot \vec v_n) \right]^2 . \nn
\eea
Now we use the steady state condition to simplify this expression. 
We write $\vec q_{(j,i)} \cdot \vec v_n =  \vec v_{n,j} -  \vec v_{n,i}$, where $\vec v_{n,j}$ denotes the $j$th component of the vector $\vec v_n$ and we obtain
\bea
    && \sum_{i,j = 1}^N A_{i,j}  \sin(\theta_i^* - \theta_j^*) 
             (\vec q_{(j,i)} \cdot \vec v_n) \nn \\
  && \qquad =  \sum_{j = 1}^N \vec v_{n,j} 
            \underbrace{ \sum_{i=1}^N   A_{i,j}  \sin(\theta_i^* - \theta_j^*)  }_{= -\omega_j/K} \nn \\
  && \qquad \qquad \qquad - \sum_{i = 1}^N \vec v_{n,i} 
            \underbrace{ \sum_{i=j}^N  A_{i,j}  \sin(\theta_i^* - \theta_j^*)  }_{= +\omega_i/K}  \nn \\
  && \qquad = -\frac{2}{K} \, \vec v_n \cdot \vec \omega \, .
\eea
The derivative of the order parameter then becomes
\bea
   \frac{\mathrm{d} r_{\rm uni}}{\mathrm{d} K} =  
    \frac{2}{K^2\sum_{i=1}^N k_i} \sum_{n=2}^N  \frac{1}{-\lambda _n}
               (\vec v_n \cdot \vec \omega )^2 , \nn
\eea
finishing the proof of Lemma 2.
\end{proof}

For any stable steady state we have $\lambda_n < 0$ for all
$n \in \left\{2,\ldots, N \right\}$ such that the slope is non-negative. It can become 
zero only if $\vec v_n \cdot \vec \omega = 0$ for all $n \in \left\{2,\ldots, N\right\}$. 
As the eigenvectors form an orthonormal basis this would imply that 
$\vec \omega$ is parallel to $\vec v_1$. As we assume $\sum_j \omega_j = 0$
this is only possible if $\vec \omega = \vec 0$ and we have $\mathrm{d} r_{\rm uni} / \mathrm{d} K > 0$ for $K > K_{c2}$.

Finally, as $K \rightarrow K_{c2}^+$ from above the phase locked state becomes unstable with $\lambda_2 \rightarrow 0$. With the assumption \mbox{$\vec \omega \cdot \vec v_2 \neq 0$} it follows that the derivative diverges, concluding the proofs for both theorems.

\section{Conclusion}

Kuramoto oscillators are the prototypical systems used to study the synchronization behavior of limit cycle oscillators. The order parameters introduced to study this synchronization capture different aspects of the transition to synchrony. None of the order parameters previously suggested for Kuramoto oscillators on complex networks describes all transitions to partial and full phase locking as well as the convergence to full synchrony in arbitrary networks.\\

Here we have proposed a universal order parameter accurately describing the phase coherence in networks of phase oscillators. This order parameter recovers the original Kuramoto order parameter for a fully connected network of oscillators. We have analytically shown that the slope of the order parameter diverges when the fully phase locked state becomes stable, accurately marking this transition even in small networks. For larger coupling strengths a monotonic increase reflects the slow convergence to complete synchrony and directly relates to the stability of the phase locked state, important, for example, for applications to power grid models where a fully phase locked state is required for stable operation.\\

\acknowledgments

We gratefully acknowledge support from the G\"ottingen Graduate School for Neurosciences and Molecular Biosciences (DFG Grant GSC 226/2 to M.S.), 
the Helmholtz Association (grant no.~VH-NG-1025 to D.W.),
the German Federal Ministry of Education and Research 
(BMBF grant no.~03SF0472B and 03SF0472E to M.T. and D.W.),  
and the Max Planck Society to M.T.



\end{document}